\documentclass[11pt,letterpaper]{article}
\usepackage{mathpazo}
\usepackage{amsfonts}
\usepackage{amsmath,amssymb}
\usepackage{hyperref}
\hypersetup{colorlinks=true,citecolor=red}
\usepackage{cleveref}
\usepackage[noend]{algpseudocode}
\usepackage{algorithm,algorithmicx,graphicx}
\usepackage{amsthm}
\usepackage{thm-restate}
\usepackage[dvipsnames]{xcolor}
\usepackage{fullpage}
\usepackage{float}
\usepackage{fullpage}
\usepackage{xspace}
\usepackage{comment}
\usepackage[OT1]{fontenc}
\usepackage{tikz}
\usepackage{xcolor}
\usepackage{multirow}
\usepackage{multicol}
\usepackage{mathabx}
\usepackage{subfigure}

\newtheorem{theorem}{Theorem}
\newtheorem{corollary}[theorem]{Corollary}
\newtheorem{definition}[theorem]{Definition}
\newtheorem{lemma}[theorem]{Lemma}

\newcommand{\abs}[1]{{\left| #1 \right|}}
\newcommand{\set}[1]{\left\{#1\right\}}
\newcommand{\eps}{\varepsilon}
\newcommand{\tp}[1]{\left(#1\right)}
\newcommand{\sqtp}[1]{\left[#1\right]}

\newcommand{\poly}{\mathrm{poly}}
\newcommand{\polylog}{\mathrm{polylog}}

\newcommand{\congest}{\mathsf{congestion}}
\newcommand{\load}{\mathsf{load}}
\newcommand{\rload}{\mathsf{rload}}
\newcommand{\flow}{\mathsf{flow}}

\newcommand{\LargeVal}{\mathsf{LARGE}}
\newcommand{\SmallVal}{\mathsf{SMALL}}

\renewcommand{\P}{\mathsf{P}}
\newcommand{\NP}{\mathsf{NP}}
\newcommand{\ZPP}{\mathsf{ZPP}}

\newcommand{\DTIME}{\mathsf{DTIME}}

\newcommand{\calD}{\mathcal{D}}
\newcommand{\map}{\mathcal{M}}
\newcommand{\height}{\mathsf{height}}
\newcommand{\Prob}{\textbf{Pr}}
\newcommand{\Exp}{\mathbb{E}}

\newcommand{\cmid}{:\, }

\def\+#1{\mathcal{#1}}
\def\=#1{\mathbb{#1}}

\title{Survivable Network Design Revisited: Group-Connectivity}

\author{
Qingyun Chen \thanks{University of California, Merced: \texttt{qingyun.chen152@gmail.com}}
\and
Bundit Laekhanukit \thanks{Shanghai University of Finance and Economics: \texttt{lbundit+sufe@gmail.com}}
\and
Chao Liao\thanks{Shanghai Jiao Tong University. Email: \texttt{chao.liao.95@gmail.com}}
\and
Yuhao Zhang\thanks{Shanghai Jiao Tong University. Email: \texttt{zhang\_yuhao@sjtu.edu.cn}}
}
\date{}
\begin{document}

\maketitle

\begin{abstract}
    In the classical survivable network design problem (SNDP), we are given an undirected graph $G=(V,E)$ with costs on edges and a connectivity requirement $k(s,t)$ for each pair of vertices. The goal is to find a minimum-cost subgraph $H\subseteq V$ such that every pair $(s,t)$ are connected by $k(s,t)$ edge or (openly) vertex disjoint paths, abbreviated as EC-SNDP and VC-SNDP, respectively. The seminal result of Jain [FOCS'98, Combinatorica'01] gives a $2$-approximation algorithm for EC-SNDP, and a decade later, an $O(k^3\log n)$-approximation algorithm for VC-SNDP, where $k$ is the largest connectivity requirement, was discovered by Chuzhoy and Khanna [FOCS'09, Theory Comput.'12]. While there is a rich literature on point-to-point settings of SNDP, the viable case of connectivity between subsets is still relatively poorly understood.

    This paper concerns the generalization of SNDP into the subset-to-subset setting, namely Group EC-SNDP. We develop the framework, which yields the first non-trivial (true) approximation algorithm for Group EC-SNDP. Previously, only a bicriteria approximation algorithm is known for Group EC-SNDP [Chalermsook, Grandoni, and Laekhanukit, SODA'15], and a true approximation algorithm is known only for the single-source variant with connectivity requirement $k(S,T)\in\{0,1,2\}$ [Gupta, Krishnaswamy, and Ravi, SODA'10; Khandekar, Kortsarz, and Nutov, FSTTCS'09 and Theor. Comput. Sci.'12].
\end{abstract}

\section{Introduction}
\label{sec:intro}

In the {\em survivable network design} problem (SNDP), we are given a graph $G=(V,E)$ with non-negative edge-costs $c:E\rightarrow\=R_{\ge 0}$ and a connectivity requirement between each pair of vertices $k:V\times V\rightarrow\=Z_{\ge 0}$. A pair of vertices $(s,t)$ with $k(s,t)>0$ is called a {\em demand-pair}, and a vertex $s$ that has a positive demand toward some vertex is called a {\em terminal}.
The goal in SNDP is to find a minimum-cost subgraph $H\subseteq G$ that has $k(s,t)$ edge (or vertex) disjoint paths connecting every demand-pair $(s,t)$.
As SNDP captures the design of communication networks that can operate under failure conditions, this problem has been a focus of attention for many decades since its initial study in the late '60s \cite{SteiglitzWK1969}.
There have been many variants of SNDP, e.g., edge-connectivity SNDP (EC-SNDP), vertex-connectivity SNDP (VC-SNDP), where the goal is to connect demand-pairs by edge-disjoint paths and (openly) vertex disjoint paths, respectively.
A number of studies have been devoted to SNDP, culminating in the discovery of a $2$-approximation algorithm for EC-SNDP by Jain \cite{Jain01}, and an $O(k^3\log n)$-approximation algorithm for VC-SNDP by Chuzhoy and Khanna \cite{ChuzhoyK12}, where $k$ is the maximum connectivity requirement.

While EC-SNDP and VC-SNDP are decent models that capture many difficulties in designing a highly reliable network, these models do not address network design beyond {\em point-to-point}. Specifically, the classical setting of SNDP concerns only the survivability of communication between pairs of nodes in a network. However, in many applications, e.g., multicasting over an overlay network, distributed data center, and global routing in VLSI design, communications are generally taken place between two groups of nodes rather than point-to-point. These relatively modern applications require the generalization of SNDP where survivability is required in the communication between two communities or, in other words, in the {\em group-to-group} setting.
We call the later model the {\em group-connectivity survivable network design problem} or {\em Group-SNDP} (resp., Group EC-SNDP for edge-connectivity).

Similar to the point-to-point network design, where the models are derived from the {\em Steiner tree} problem. The basic building block of community-to-community network design lays on the classical {\em group Steiner tree} problem (GST), where we are given subsets of vertices, called {\em groups}, and the goal is to find a minimum-cost tree that spans all the groups, i.e., the tree must connect to at least one vertex from each group.
The generalization of GST into the fault-tolerant settings, namely the {\em fault-tolerant GST} or {\em $k$-edge-connected} GST ($k$-EC-GST), have been studied in \cite{KhandekarKN12,GuptaKR10,ChalermsookGL15}, and the generalization into pairwise group-connectivity, namely Group-SNDP (a.k.a., {\em survivable set-connectivity}) has been studied in \cite{ChalermsookGL15}.
These studies culminated in a polylogarithmic approximation algorithm for $k$-EC-GST when $k=2$ \cite{GuptaKR10}, and a bicriteria approximation algorithm for the general edge-connectivity demands of Group SNDP (Group EC-SNDP) \cite{ChalermsookGL15}. Nevertheless, to the best of our knowledge, there was no known ``true'' (non-trivial) approximation algorithm for neither Group EC-SNDP nor $k$-EC-GST for $k\geq 3$.
On the negative side, Group EC-SNDP or even $k$-EC-GST is known to admits no $k^{\sigma}$-approximation algorithm, for some fixed constant $0 < \sigma < 1$ unless $\P=\NP$ \cite{ChalermsookGL15}, and when $k$ is very large the problem admits no $O(q/\log q)$-approximation algorithm unless $\NP=\ZPP$ \cite{LiaoCLZ2022}.

This paper revisits the generalization of survivable network design problem in the group-connectivity setting, namely Group EC-SNDP. To be formal, in the Group EC-SNDP we are given an undirected graph $G=(V,E)$ with non-negative edge-costs, a collection of pairs of subsets of vertices $(S_1,T_1),\ldots,(S_q,T_q)\subseteq V\times V$ with prescribed connectivity requirements $k_i \in\=Z_{\ge 0}$ for $i=1,2,\ldots,q$. The goal is to find a minimum-cost subgraph $H\subseteq G$ that has $k_i$ edge-disjoint paths between every pair of subsets $(S_i,T_i)$.

We present the first non-trivial (true) approximation algorithm for Group EC-SNDP, thus solving a long-standing open problem in the area of network design.
The approximation guarantee of our algorithm is $O(k^2\log k \log^6 n (k\log n+ \log q)))$, where $k$ is the largest connectivity requirement.
In particular, it is {\bf $O(k^3\log k \log^7 n)$} when $q=n^{O(k)}$, which resemblances the best known approximation ratio of $O(k^3\log n)$ for VC-SNDP upto the polylogarithmic term. Notice that the number of demand-pairs $q$ can be super-polynomial in $n$, and in this case, when $q\geq n^{\omega(k)}$, the approximation ratio of our algorithm becomes $O(k^2\log k\log^6n\log q)$.

\subsection{Related Work}
\label{sec:intro:related}

The study of survivable network design problems was initiated in the late '60s \cite{SteiglitzWK1969}.
Since then many variants and generalizations have been modeled to capture a wide range of situations, e.g., edge-connectivity (EC-SNDP), vertex-connectivity (VC-SNDP), and element-connectivity (Elem-SNDP), subset-to-subset connectivity (Group-SNDP) and connectivity in directed graphs (Directed-SNDP).
Please see \cite{KerivinM05-Survey} and \cite{Nutov18a-Survey} for references therein.


EC-SNDP is the most well-studied among the problems in the class of survivable network design. It has been extensively studied in the '90s (see, e.g., \cite{WilliamsonGMV93,GoemansGPSTW94,Jain01}), culminating in the discovery of a $2$-approximation algorithm via iterative rounding method in the breakthrough result of Jain \cite{Jain01}. The same technique generalizes to Elem-SNDP in the work of  Fleischer, Jain, and Williamson \cite{FleJW06}, thus giving a factor-two approximation algorithm for this variant as well.
In contrast to EC-SNDP and Elem-SNDP, researchers have been struggling in developing approximation algorithms for VC-SNDP.
To the best of our knowledge, there is only one known non-trivial approximation algorithm for VC-SNDP, which was discovered decades later in the work of Chuzhoy and Khanna \cite{ChuzhoyK12}, giving an $O(k^3\log n)$-approximation algorithm to the problem, where $k$ is the largest connectivity requirement.
One reason for the difficulty of VC-SNDP is due to the hardness derived from the Label-Cover problem.
Assuming $\NP\not\subseteq \DTIME(n^{\text{polylog}(n)})$, almost all reasonable approximation ratios (i.e., $2^{\log^{1-\epsilon}n}$, for $\epsilon>0$) have been ruled out by the work of Kortsarz, Krauthgamer, and Lee \cite{KortsarzKL04}, and the approximation ratios independent of $k$ are ruled out by the work of Chuzhoy, Khanna, and Chakraborty \cite{ChakrabortyCK08}.
A more refined approximation of hardness was later shown in \cite{Lae2014}.
Nevertheless, a special case of VC-SNDP called $k$-connected (spanning) subgraph problem admits constant factor approximation ratios for almost every parameter \cite{CheriyanV14,FukunagaNR15,Nutov20-kVCSS} except in the large connectivity regime, which still have polylogarithmic factor \cite{FakcharoenpholL12,Nut2012}.
Please see the recent comprehensive survey on VC-SNDP by Nutov \cite{Nutov18a-Survey}.

On directed graphs, both EC-SNDP and VC-SNDP are equivalent as there are polynomial-time reductions from one to the other. Thus, we denote them simply by Directed-SNDP.
The problem seems to be much more difficult as it admits almost no approximation ratios in $n$ due to the work of Dodis and Khanna \cite{DodS1999}, which gives a reduction from the notorious Label-Cover problem to Directed-SNDP. In fact, the hardness result holds even when all the connectivity requirements are $\{0,1\}$. The bounds have been improved in the subsequent works to $n^{o(1)}$ under Gap-ETH in the work of Dinur \cite{Dinur16} and has also been refined in \cite{CheLNV2014,Lae2014,DinM18,Man2019,LiaoCLZ2022}. To date, we know that even in the special case of single-source $k$-connectivity, the problem is at least as hard as the Label-Cover problem \cite{CheLNV2014}, and in the very recent work, Liao, Chen, Laekhanukit and Zhang showed that Directed-SNDP admits no non-trivial approximation algorithms unless $\NP=\ZPP$. To be more precise, Directed-SNDP admits neither $o(q/\log q)$ nor $o(2^{k/2}/k)$ approximation algorithms unless $\NP=\ZPP$, where $q$ is the number of (positive) demand-pairs. Assuming the {\em Strongish Planted Clique Hypothesis} \cite{ManRS21}, even $o(q)$-approximation algorithm has been ruled out \cite{LiaoCLZ2022}. Despite the difficulties, some special cases of Directed-SNDP admits a polylogarithmic approximation factor \cite{ChaLWZ2020,Nutov21-QkDST}.

    The Group SNDP is far less understood than other models for both edge and vertex connectivity variants. The study of the classical group Steiner tree problem is shown in \cite{GargKR00} and the generalization to the special case of single-source two-edge-connectivity was studied in \cite{KhandekarKN12,GuptaKR10}, resulting in polylogarithmic approximation algorithms for these two cases. A polylogarithmic approximation algorithm for general demands is also known in a restricted setting of low treewidth graphs \cite{ChalermsookDELV18}. However, prior to our paper, only a bi-criteria approximation algorithm is known for Group EC-SNDP on general graphs \cite{ChalermsookGL15}. Similar to Directed-SNDP, when the connectivity requirements are very large, Group EC-SNDP admits no non-trivial approximation. That is, no $o(q/\log q)$-approximation algorithm exists unless $\NP=\ZPP$, and no $o(q)$-approximation algorithm exists unless the Strongish Planted Clique Hypothesis is false \cite{LiaoCLZ2022}. When focusing on the hardness ratio in terms of $k$, it is shown in \cite{ChalermsookGL15} (combined with \cite{Lae2014} and the improvement in \cite{Man2019}) that the approximation hardness is $k^{1/5-\epsilon}$, for $\epsilon>0$, assuming $\NP=\ZPP$.

\subsection{Result and Contribution}
\label{sec:intro:results}

As mentioned earlier, the main result in this paper is the first non-trivial (true) approximation algorithm for Group EC-SNDP.

\begin{theorem}[Main Result]
\label{thm:main}
The group edge-connectivity survivable network design problem admits a polynomial-time (randomized) $O(k^2\log k \log^6 n (k\log n+ \log q))$-approximation algorithm, where $k$ is the largest pairwise connectivity requirement. In particular, the approximation ratio becomes $O(k^3\log k\log^7 n)$ when $q \le n^{O(k)}$.
\end{theorem}

The key ingredient is a capacity-based probabilistic tree-embedding by R\"{a}cke \cite{Rac2008}, which is also used in the previous work by Chalermsook, Grandoni and Laekhanukit \cite{ChalermsookGL15}. However, Chalermsook~et~al. were not able to derive a true approximation algorithm from the capacity-based tree-embedding because such a mapping is ``lossy''. More precisely, their algorithm is a randomized LP-rounding algorithm that runs on a tree distribution, which is an embedding of a fractional solution.
However, due to the distortion, although there is an $(S_i,T_i)$-flow of value $k_i$ in the tree distribution, when we map it back to the original graph, we can guarantee only a flow of value at least $\Omega(k/\beta)$, where $\beta=O(\log n)$ is the capacity-distortion of R\"{a}cke's tree-embedding. Consequently, the algorithm in \cite{ChalermsookGL15} is only guaranteed to output a solution with connectivity at least $\Omega(k/\log n)$.

To circumvent the issue of lossy embedding, we are required to invoke two techniques. Firstly, we apply a weight-updating technique similar to that in the multiplicative weight update method. (See the survey by Arora, Hazan, and Kale \cite{AroraHK12} for more details.)
Every time we add new edges to the partial solution, we update the weights (capacities) of the bought edges by scaling their capacities down so that the distortion on these edges has only mild effects toward distortion. However, the weight-update is still insufficient for us to reach the desired connectivity. This is because a cut consisting of bought edges alone may have its capacities scaled down too much so that the desired connectivity in the tree cannot be guaranteed.
Thus, the connectivity issue remains and we need an additional ingredient.
To this end, we observe that all the cuts in the graph are present in the embedded trees in the form of leaves-to-leaves paths. Hence, to reach the desired connectivity, we can simply connect all of them simultaneously whenever they are {\em connectivity deficient}.
It is easier said than to be done as the number of tree-demand-pairs that we need to {\em cover} (that is, adding paths to satisfy the connectivity requirement) is, in general, exponential on the number of vertices. Thus, even if we have a logarithmic approximation algorithm for the underlying network design problem on trees (which is called {\em subset-connectivity} in \cite{ChalermsookGL15}), it can only give a polynomial approximation ratio.
Nevertheless, we are able to show that the number of tree-demand-pairs needed to be covered is at most $q\cdot n^{2k}\cdot 2^{\poly(k)\beta}$, which is $qn^{\poly(k)}$ when $\beta=O(\log n)$.
As a consequence, a polylogarithmic approximation algorithm for the subset-connectivity problem on trees in \cite{ChalermsookGL15} implies a $\poly(k)\polylog(n)$ approximation algorithm for Group EC-SNDP.

The techniques that we developed yield a framework that turns a ``lossy'' probabilistic capacity mapping into a ``lossless'' network design algorithm. Hence, it can be applied to a more general setting whenever a lossy capacity-based probabilistic tree-embedding exists, provided that the distortion is $\beta < n^{1/3}$.
The general form of our result is as follows.

\newpage

\begin{theorem}[Network Design via Lossy-Embedding]
\label{thm:general-result}
Suppose there exists a probabilistic capacity mapping that maps a capacitated directed or undirected $n$-vertex graph $G$ into a distribution of tree such that
\begin{itemize}
    \item The congestion in expectation is $\beta$;
    \item The height of all tree is $O(\log (nC))$, where $C$ is the ratio of the largest to smallest capacity.
\end{itemize}
Then there exists a randomized $O(\beta^2 k^2\log k \log^4 n(k\log n + k\beta + \log q))$ approximation algorithm, where $k$ is the largest pairwise connectivity requirement and $q$ is the number of demand-pairs.
\end{theorem}

In particular, 
if $\beta = O(1)$, then the ratio becomes $O(k^2\log k\log^4 n(k\log n + \log q))$.
If $\beta=O(\log n)$, then the ratio becomes $O(k^2\log k\log^6 n(k\log n + \log q))$.
If $\beta = O(n^c)$ where $c<1/3$, then the ratio becomes $O(\beta^2 k^2\log k \log^4 n (k\beta + \log q))$.

\Cref{thm:general-result} allows us to derive a slightly better approximation ratio on special classes of graphs, e.g., graphs with bounded pathwidth,  bandwidth or cutwidth \cite{BorradaileCEMN20} and $k$-outer planar graphs \cite{Yuval11}.
Note that here we exploit the equivalence between distance and capacity-based probabilistic tree-embedding observed by Andersen and Feige \cite{AndersenFeige09}.

\begin{theorem}
Consider the Group EC-SNDP. There exist polynomial-time (randomized) approximation algorithms with approximation ratios:
\begin{itemize}
    \item $O(b^4 \cdot k^2 \log k \log^4 n (k\log n + \log q))$ for Group EC-SNDP on graphs with $b$-bounded pathwidth (resp., bandwidth and cutwidth), where $b$ is a constant.
    \item $O(c^{2k}\cdot k^2\log k \log^4 n(k\log n + kc^k + \log q))$ for Group EC-SNDP on $k$-outer planar graphs, where $c$ is a universal constant.
\end{itemize}
\end{theorem}

Lastly, we remark that a capacity-based probabilistic tree-embedding on {\em $\alpha$-balanced directed graphs} exists as shown in the work of Ene, Miller, Pachocki and Sidford \cite{EneMPS16}. However, their work is pertained to the single-source congestion minimization problem, which is not clear whether the technique fits in our framework. If it is applicable, perhaps with some modification, then our framework will imply a non-trivial approximation algorithm for the {\em singe-source directed $k$-edge-connectivity} problem (also called the {\em $k$-edge-connected directed Steiner tree} problem in \cite{GraL2017}), for all values of $k$, on $\alpha$-balanced graphs.
 

\section{Technique: Overview and Intuition}
\label{sec:overview}

We first discuss the idea used in the work of Chalermsook, Grandoni, and Laekhanukit \cite{ChalermsookGL15}.
The authors derived a bi-criteria approximation algorithm for Group EC-SNDP by first solving the standard LP-relaxation and then embedding the fraction solution into a probabilistic distribution of trees. They then iteratively round the fractional solution using the Garg-Konjevod-Ravi rounding algorithm for the group Steiner tree problem on trees \cite{GargKR00}.

Due to the {\em congestion} (i.e., capacity-distortion) $\beta=O(\log n)$ of the R\"{a}cke's capacity-based probabilistic tree-embedding, the solution could only be guaranteed to reach a connectivity at least $\Omega(k/\log n)$. It is quite interesting that the problematic edges are those that have large LP-values, say $x_e\geq 1/\beta$. This is because the effect of distortion on the capacities of these edges are large. Surprisingly, the good case where we can reach the target connectivity is when all the edges in the optimal LP-solutions have {\em small} LP-values, i.e., $x_e< 1/\beta$.

It is very counterintuitive that we wish for an LP-solution with no large LP-values because these edges are supposed to be ``good-to-have'' as we can trivially round $x_e$ to one and simply pay a factor $O(\log n)$ in the approximation ratio. This suggests the mix of {\em tree-rounding} and {\em trivial-rounding}. However, a straightforward approach would inevitably fail because the connectivity has already been lost in the tree distribution.
To be specific, let us consider a cut $(X,V\setminus X)$ with capacity exactly $k$, i.e., a {\em tight} cut that separates a demand-pair $(S_i,T_i)$.
Suppose there are $k/\beta$ edges with capacity $1$ that have congestion $\beta$ crossing $(X,V\setminus X)$.
Then this cut will appear to have a flow of values $k$ in the tree distribution.
However, even if we buy all the $(S_i,T_i)$-paths in the support of the distribution, it would form only $O(k/\beta)$ edge-disjoint paths in the original graph.

To circumvent this issue, we study the effect of distortion of large-capacity edges and analyze in detail for which cut will be capacity-deficit when we embed it to the tree distribution.
To simplify the discussion, fix an edge-set $F$ of $\ell \leq k-1$ edges in the graph.
If a demand-pair $(S_i,T_i)$ is not yet $k$-edge-connected in the current partial solution, then removing $F$ from the solution subgraph will disconnect them.
That is, $F$ is a certificate that the graph has not yet reached the desired connectivity.
The same applies to the tree-embedding.
If we remove edges in the embedded-tree corresponding to $F$ (i.e., an edge that maps to a path containing an edge in $F$) and the pair $S_i$ and $T_i$ are disconnected, then $S_i$ and $T_i$ are not $k$-edge-connected in the original graph even if we buy all the $(S_i,T_i)$-paths in the tree distribution.
Now, observe that any edge with capacity (which is its LP-value) less than $1/4\ell\beta$ will have no effect in separating $S_i$ from $T_i$ because even if it has capacity-distortion $\beta$, and we remove $k$ of them, it cannot possibly separate $S_i$ from $T_i$ in the tree-embedding.
Hence, we may assume the edges in the graph have capacities at most $1/4\ell\beta$ by simply scaling down the capacities of edges with {\em large} LP-values, or more precisely, we cap the capacity of any edge to be at most $1/4\ell\beta$.

Next, we observe the embedded tree after removing edges corresponding to the set $F$ of $\ell$ edges.
One can see that the cut that has a lot of edges with {\em small} LP-values will not be disconnected by removing $F$.
More precisely, if the total capacities of edges with small LP-values is at least $1$, then it will have a capacity of at least $1/2$ after removing $F$, and such a condition holds in every feasible LP-solution.
As such, the cuts that are ``shattered'' in the tree-embedding are those that are already $\ell$-edge-connected in the solution subgraph.

This observation is quite interesting for us as because it means that we can view a pair of shattered components as super nodes and ask only for edge-disjoint paths connecting these components. Towards this goal, we define new tree-demand-pairs based on the cuts and the embedding tree that are {\em shattered-free}.
Then it is not hard to see that joining these subset pairs through edge-disjoint paths will imply the connectivity in the original graph.
The last obstacle lies on the fact that the number of tree-demand-pairs can be as large as $2^n$ in general. The critical point in our argument is in bounding the number of ``shattered-free'' tree-demand-pairs and showing that every tree-demand-pair has enough flow for the rounding algorithm.
Once these are all settled, our algorithm is guaranteed to connect all the tree-demand-pairs within the claimed bound.


\section{Preliminary}
\label{sec:prelim:tree-embedding}


\paragraph*{Group Edge-Connectivity Survivable Network Design.} The group edge-connected survivable network design problem (Group EC-SNDP) is defined as follows: Given an undirected graph $G=(V,E)$ with non-negative edge costs $c:E\to\=R_{\ge 0}$, and a collection of $q$ demand-pairs $(S_1,T_1),\ldots,(S_q,T_q)$ with a connectivity requirement $k_i\in\=Z_{\ge 0}$ for each demand-pair $(S_i,T_i)$, the goal is to find a minimum-cost subgraph $H\subseteq G$ such that $H$ has $k_i$-edge-disjoint paths from $S_i$ to $T_i$ for every $i\in[q]$.

One may assume w.l.o.g. that the connectivity requirements are uniform.
To see this, let $k=\max_{i\in [q]}k_i$.
We add $k$ auxiliary edges with zero-cost to the graph, say $(a_{1},b_{1}),(a_{2},b_{2}),\ldots,(a_{k},b_{k})$.
Then, for each demand-pair $(S_i,T_i)$ with requirement $k_i<k$, we add $a_j$ to $S_j$ and $b_j$ to $T_j$ for $j=1,2,\ldots,k-k_i$.
This adds $k-k_i$ independent edges joining $S_i$ and $T_i$ and only increases the size of the instance by $k$ even when $q$ is superpolynomial in $n$.
It is not hard to see that any feasible solution to the modified graph with uniform requirement $k$ induces a feasible solution to the original graph with the same cost and vice versa.
As such, our discussion will focus on the uniform case.

Sometimes we consider more than one graph and, to clarify the notations, we use $V(\cdot)$ or $E(\cdot)$ to mean the vertex and edge set of a graph, respectively.
For a set $X\subseteq V$ of vertices, we denote by $\delta_G(X)$ the edge set between $X$ and $V\setminus X$.
The standard cut-based LP-relaxation of Group EC-SNDP is as follows.
For notational convenience, let $x(F)=\sum_{e\in F}x_e$ for a set of edges $F\subseteq E$.

\begin{equation}
    \label{eqn:lp}
    \begin{aligned}
    \min && \sum_{e\in E(G)}c_ex_e\\
    \text{s.t.}
      && \sum_{e\in\delta_G(X)}x_e \geq k
        && \forall i\in[q], \forall X:T_i\subseteq X \subseteq V\setminus S_i\\
      && 0 \leq x_e \leq 1
        && \forall e\in E(G)
    \end{aligned}
\end{equation}

\paragraph{Connectivity Augmentation.} Instead of directly rounding the fractional solution for Group EC-SNDP, we apply connectivity augmentation to design an algorithm. This approach is widely used in literature; see, e.g., \cite{WilliamsonGMV93,GoemansGPSTW94,KortsarzN05,CheriyanVV03,FakcharoenpholL12,Nutov12,Nutov18a-Survey}. In particular, we assume that we are given a partial solution, which is a subgraph $H_\ell=(V,E_\ell)$ of $G$ such that all $(S_i,T_i)$ are $\ell$-edge-connected.
The goal is to choose a minimum-cost edge-set $E^+\subseteq E\setminus E_{\ell}$ so that in $(V,E_\ell\bigcup E^+)$ all demand-pairs $(S_i,T_i)$ are $(\ell+1)$-edge-connected.
Starting with the trivial case that $\ell=0$, after $k$ rounds of augmentation, all $(S_i,T_i)$ become $k$-edge-connected.
It can be proved using LP-scaling technique that the connectivity augmentation only incurs an extra $O(\log k)$ factor in the approximation ratio.
See \cite{GoemansGPSTW94} for an example of analysis.

To this aim, it suffices to set $c_e=0$ for all $e\in E_\ell$ and set $k=\ell+1$ in the LP formulation \eqref{eqn:lp}.
Then \eqref{eqn:lp} becomes the standard cut-based LP for the augmentation problem.
In particular, we have the following lemma.

\begin{lemma} [Connectivity Augmentation]
\label{lem:augmentation}
Suppose there exists an $\alpha(n,k)$-approximation algorithm for the connectivity augmentation problem w.r.t. its LP-solution, i.e., it produces a feasible solution with cost at most $\alpha(n,k)\cdot z$, where $z$ is the cost of the optimal fractional LP-solution.
Then there exists an $O(\log k \cdot \alpha(n,k))$ for Group EC-SNDP.
\end{lemma}


\subsection{R\"{a}cke's Probabilistic Tree-Embedding}

Our algorithm uses a tree-embedding with congestion $\beta$ as a black box.
The main result is proved by the fact that R\"{a}cke's probabilistic tree-embedding has congestion $\beta=O(\log n)$ in expectation. In this section, we give the formal definition of R\"{a}cke's probabilistic tree embedding. Let $G=(V,E)$ with capacity $\widetilde{x}:E\rightarrow\=R_{\ge 0}$. Denote a tree embedding of $(G,\widetilde{x})$ by $(\+T,\+M,y)$, where $\+T$ is a tree, $\+M$ is a mapping from $V(\+T)\bigcup E(\+T)$ to $V(G)\bigcup 2^{E(G)}$, and $y$ is the corresponding capacity function on the edges of $\+T$. For each node $v$ of $\+T$, $\+M(v)$ is some vertex in $G$.
Particularly, $\+M$ induces a one-to-one mapping between the leaves of $\+T$ and the vertices of $G$.
For each edge $f=(u,v)$ in $\+T$, its capacity is $y(f)=\sum_{e\in\delta_G(X)}\widetilde{x}_e$, where $(X,V\setminus X)$ is a partition induced by the leaves of $\+T-f$.
In addition, $\+M$ maps $f$ to a path in $G$ between $\+M(u)$ and $\+M(v)$.

We will use following notations in our analysis repeatedly. For each vertex $v$ of $G$, $\+M^{\text{-1}}(v)$ is the corresponding leaf of $\+T$.
Let $\+M^{\text{-1}}(X)=\bigcup_{v\in X}\+M^{\text{-1}}(v)$ for $X\subseteq V(G)$.
For each edge $e$ in $G$, let $\+M^{\text{-1}}(e)=\set{f\in E(\+T) : e \in\+M(f)}$ and let $\+M^{\text{-1}}(F)=\bigcup_{e\in F}\+M^{\text{-1}}(e)$ for $F\subseteq E(G)$.

\paragraph{Congestion.}
For each edge $e\in E(G)$,  the load of $e$ on a tree $\+T\in\calD$ is defined as the sum of the capacities of edges $\+M^{\text{-}1}(e)$, i.e., $\load(e) = \sum_{f\in\+M^{\text{-}1}(e)} y(f)$. Let the relative load be $\rload(e)=\load(e)/\widetilde{x}_e$. For a tree embedding $(\+T,\+M,y)$, the \emph{congestion} of $G$ is $\max_{e\in E(G)}\{\rload(e)\}$. Finally, we define $\beta$ to be the expected congestion:
$$
\beta:= \max_{e\in E(G)}\Exp_{\+T\sim\calD}[\rload(e)],
$$
where $\+D$ is the probability distribution on trees given by R\"acke's tree-embedding.

\begin{theorem}[\cite{Rac2008}]
\label{thm:racke08-embedding}
There exists a probabilistic embedding of an $n$-vertex graph $G$ with edge capacities into a tree with expected congestion at most $\beta=O(\log n)$.
\end{theorem}

The embedding also guarantees some properties about flows.
Let $\flow_G^{\widetilde{x}}(A,B)$ denote the maximum flow from $A\subseteq V(G)$ to $B\subseteq V(G)$ in $G$ under capacity $\widetilde{x}$. For any pair of disjoint subsets $A,B\subseteq V(G)$, the value of the maximum flow between $A$ and $B$ in $\+T\in\calD$ is at least that of the maximum flow between $A$ and $B$ in $G$, i.e., $\flow_G^{\widetilde{x}}(A,B) \leq \flow_{\+T}^y(A,B).$ Since each edge has congestion at most $\beta$ in expectation, any flow that can be routed on the tree distribution can be routed in the original graph with a loss of a factor of $\beta$, i.e.,  $\flow_G^{\widetilde{x}}(A,B) \geq \frac{1}{\beta}\cdot\mathbb{E}_{\+T\sim\calD}[\flow_{\+T}^y(A,B)].$

It was also known that the height of the tree in R\"acke's tree distribution can be bounded by the largest to smallest ratio of the capacity, and the number of trees in the support of the Racke's tree distribution is at most  on $O(\poly(n))$.

\begin{lemma} [\cite{Rac2008}]
\label{lem:racke-height}
All the trees $\+T$ in the support of the R\"acke's tree distribution have height $O(\log(nC))$, where C is the ratio of the largest to smallest capacity in $\widetilde{x}$.
\end{lemma}

\section{Algorithm}
\label{sec:approx-algo}

Now we describe the approximation algorithm for the connectivity augmentation problem.
First, we solve the augmentation LP by setting $k=\ell+1$ and $c_e=0$ for all $e\in E_\ell$ in \eqref{eqn:lp}.
Then, we obtain an LP solution $\{x_e\}_{e\in E(G)}$.
We assume $x_e=1$ for all edges $e\in E_\ell$.
Let $\beta$ be the congestion parameter of the tree-embedding. We then define two subsets of edges based on $\beta$ and $\ell$:
\begin{align*} 
    \LargeVal &=\left\{e\in E(G)\cmid x_e \geq \frac{1}{4\ell\beta}\right\},\\
    \SmallVal &=\left\{e\in E(G)\cmid x_e < \frac{1}{4\ell\beta}\right\}.
\end{align*}

Our algorithm buys edges in two different rounds. In the first step, we directly buy all edges in $\LargeVal$ to the solution subgraph $H_{\ell+1}$.
This incurs a factor $O(\ell\beta) = O(k\log n)$ in the approximation ratio.



After that, if $H_{\ell+1}$ is already $(\ell+1)$-connected, i.e., there are $\ell+1$ edge-disjoint paths from every $S_i$ to $T_i$, then we are done. If not, we continue to run a dependent rounding algorithm on the tree-embedding. We build the R\"acke's tree (a distribution of trees) with the capacity defined as follows. We will cap the capacity of the edges in $\LargeVal$ to be exactly $\frac{1}{4\ell\beta}$, while keeping the same value for edges $e$ with $x_e<\frac{1}{4\ell\beta}$. We also omit the fractional solutions that are smaller than $\frac{1}{2n^2}\cdot\frac{1}{4\ell\beta}$, to control the height of the trees.
Formally, we set the capacity of each edge $e\in E(G)$ as
\begin{equation}
\widetilde{x}_{e} = \left\{\begin{array}{ll}
  \frac{1}{4\ell\beta} &
  \text{if $e$ is in $\LargeVal$,}\\
  0                      & x_e < \frac{1}{2n^2}\cdot\frac{1}{4\ell\beta},\\
  x_e                      & \text{otherwise}.

\end{array}\right.
\label{eqn:capacity}
\end{equation}

By the setup of $\widetilde{x}$, the height of tree is bounded by $O(\log n)$.
\begin{lemma}
\label{lem:tree-height}
Let $\+D$ be the R\"acke's tree distribution for $(G,\widetilde{x})$, the height of each $\+T\sim \+D$ is at most $O(\log n)$.
\end{lemma}
\begin{proof}
It directly follows from \Cref{lem:racke-height} since the largest to smallest ratio w.r.t. $\widetilde{x}$ is $C\leq 2n^2$.
\end{proof}

Now we are ready to state our algorithm.
See \Cref{alg:main}.
The algorithm first computes an LP solution and set up a basic solution graph $H_{\ell+1}=(V,E_{\ell+1})$ by setting $E_{\ell + 1}=\LargeVal$.
Then it tries to add edges to $E_{\ell + 1}$ to make $H_{\ell+1}$ $(\ell+1)$-connected in $\tau$ rounds.
In each round, it samples a tree from R\"acke's tree distribution.
Then it samples certain edges in the tree by the rounding algorithm from Grandoni-Chalermsook-Laekhanukit \cite{ChalermsookGL15} and appends the corresponding edges in the graph $G$ of the sampled edges on the tree to $E_{\ell + 1}$.

We will present the main flow of the proof in \Cref{sec:analysis}. Also, for completeness, we will formally present the subroutine of tree rounding and its analysis in \Cref{sec:roundtree}.

\begin{algorithm}[H]
\caption{Algorithm for the Augmentation Problem.}
\label{alg:main}
\begin{algorithmic}[1]
    \Require  An $\ell$-connected subgraph $H_\ell=(V,E_\ell)$ ($\ell$-connected for each demand-pair).
    \State Define $c_e=0$ for those $e\in E_\ell$.
    \State Solve $x$ to be the fractional solution to the LP in \eqref{eqn:lp} with $k=\ell+1$. We assume all $x_e=1$ when $c_e=0$.
    \State Buy $\LargeVal$ edge set where $x_e \geq \frac{1}{4 \ell \beta}$ to $H_{\ell+1}$.
    \State Set the capacity $\widetilde{x}$ of edges following Eqn~\eqref{eqn:capacity}.
    \State Compute R\"{a}cke's tree distribution $\calD$ for $(G,\widetilde{x})$.
    \For {$\tau_1$ rounds}
        \State Sample a tree $(\+T,\+M,y)$ from $\calD$.
        \For {$\tau_2$ rounds}
            \State Append TreeRounding $(G,\+T)$ to $H_{\ell+1}$.
        \EndFor
    \EndFor
    \State return $H_{\ell+1}$.
\end{algorithmic}
\end{algorithm}

\section{Analysis: Component-level Paths and $(\ell+1)$-Connectivity}
\label{sec:analysis}
\paragraph*{CGL's Tree Rounding as a Black Box.}
Let us discuss the CGL's tree rounding subroutine in detail.
The subroutine keeps sampling trees from the distribution.
In each iteration, it randomly buys some edges in the tree with expected cost of $O(\text{polylog}(n)) \cdot \sum_{e\in E} c(e)x_e$.
As a result, each path in the tree that carries one unit of flow (or at least some fixed constant) has a constant probability to be selected.
Then if we have $\psi$ number of demanding pairs with constant flow to connect.
We need to suffer $O(\log(\psi))$ rounds to connect all of them in constant probability, which concludes in the $O(\text{polylog}(n) \cdot \log(\psi))$ approximation ratio.
To be more precise, we will need the following lemma in our proof. For completeness, we leave the proof of the lemma in \Cref{sec:roundtree}.

\begin{restatable}[CGL's Tree Rounding]{lemma}{treerounding}
\label{lem:cgl}
Suppose that there is a tree $\+T$ with $\height(\+T)=O(\log n)$ and an edge set $E'\subseteq E(\+T)$ that supports a flow of value at least $f$ between two vertex sets $A$ and $B$. If only edges from $E'$ can be selected by $\text{TreeRounding}(G,\+T)$, then it connects $A$ to $B$ with constant probability $\phi$ with cost of $O(\frac{1}{f} \cdot \beta \cdot \log^3 n) \cdot \sum_{e\in E} c_e x_e$.
\end{restatable}

In this part, we illustrate the big picture of our analysis. To show our main idea clearly, we assume that there is no edge with $x_e<\frac{1}{2n^2}\cdot \frac{1}{4\ell\beta}$. Although we need to scale down those tiny edges to control the height of R\"acke's tree, it only incurs a constant factor loss of the approximation ratio, and we will consider it in the complete proof in the following subsections.

The goal of our analysis is to formulate a sufficient condition for the $\ell+1$ connectivity, which consists of a small number of demanding pairs $(A,B)$ so that we can connect all of them with not too many iterations. In \Cref{alg:main}, after buying $\LargeVal$ edges, we have a subgraph $H_{\ell+1}$ that is not yet $(\ell+1)$-connected. We use $H$ to mean the state of $H_{\ell+1}$ at that moment, which contains all $\LargeVal$ edges. We will use $H_{\ell+1}$ to mean the final subgraph we return. Because $H$ is only $\ell$-connected, there are some edge-sets $F\subseteq E(H)$ of $\ell$ edges such that $E(H) \setminus F$ has no path from $S_i$ to $T_i$ for some $i$. To show the $(\ell+1)$-connectivity, we need to prove that, for any  edge-set $F$ of size $\ell$, there is a path in $H_{\ell+1} \setminus F$ for those unconnected pairs $(S_i,T_i)$ in $H \setminus F$.


A tree $\+T$ is \emph{good} for the cut $F$ if the load of $F$ on $\+T$ is at most $1/2$, i.e., $y(\map^{-1}(F)) \leq 1/2$. We prove that there is at least a tree in the support of the tree distribution that is \emph{good} for $F$. For a fixed $F$, we only focus on a \emph{good} tree because $F$ can at most block a flow of $1/2$ in $\+T$ and we does not lose too much connectivity in $\+T$ if we remove $F$.

\begin{lemma}
\label{lem:good-tree-probability}
For each edge-set $F\subseteq E(H)$ of exactly $\ell$ edges, a tree sampled in $\+T \sim \+D$ is good for $F$ with probability at least $1/2$.

\end{lemma}
\begin{proof}
    Since all edges in $F\subseteq E(H_{\ell})$ are $\LargeVal$ and $|F|=\ell$, the total capacity of $F$ w.r.t $\widetilde{x}$ is $\ell \cdot \frac{1}{4\ell\beta} = \frac{1}{4\beta}$. Note that the expected congestion of every edge in $F$ is at most $\beta$ by definition, so $\sum_{e\in F} \load(e) \leq 1/4$. We have $\Exp[y(\+M^{\text{-}1}(F))]=\Exp[\sum_{e\in F} \load(e)]\leq 1/4$. By Markov's inequality, we have that 
    \begin{align*}
        \Prob \sqtp{\widebar{\+E}} &=\Prob\sqtp{y(\+M^{\text{-}1}(F))\geq 1/2} \\
                             &\leq \Prob \sqtp{y(\+M^{\text{-}1}(F))\geq 2\Exp[y(\+M^{\text{-}1}(F))]}
                             \leq 1/2.\qedhere
    \end{align*}
    
\end{proof}

The na\"ive plan is to prove that the TreeRounding subroutine can return a new path avoiding $F$ between every $(S_i,T_i)$ with constant probability directly by \Cref{lem:cgl}. Hence, we need to show that the flow in $\+T \setminus \+M^{\text{-}1}(F)$ between $S_i$ and $\+T_i$ is at least a constant. Then we are done with an approximation ratio of $O((\ell \cdot \log n  + \log q)\cdot \log^4 n)$ because we only have $O(qn^{2\ell})$ choices of $(F,i)$ pair.
The flow between $S_i$ and $T_i$ in $\+T \setminus \+M^{\text{-}1}(F)$ is at least its flow in $\+T$ minus the value of $\+M^{\text{-}1}(F)$. If we focus on a good tree $\+T$, because $F$ only blocks a flow of value $1/2$, it suffices to prove the total flow from $S_i$ to $T_i$ in $\+T$ is larger than $1/2$. However, it is not guaranteed because the connectivity from $S_i$ to $T_i$ may totally lose if we remove $F$. Let us check the following example.

\begin{itemize}
    \item (Refer to \Cref{fig:planA}) The corresponding path $\map(e)$ in $G$ of the red edge $e\in \+T$ is part of a path from $S_i$ to $T_i$ in $G$.
    If there is some edge $f\in F$ such that $f\in \+M(e)$ then the flow from $S_i$ to $T_i$ in $T\setminus \map^{-1}(F)$ becomes zero (disconnected).
    Even if $\+T$ is a \emph{good} tree, this situation can happen.
    By the feasibility of $x$ we have $x(\delta_G(S_i\bigcup\{v_1\}))\geq \ell+1$.
    However, if $\delta_G(S_i\bigcup\{v_1\}) \subseteq H_\ell$, $\widetilde{x}(\delta_G(S_i\bigcup\{v_1\}))$ can be scaled down to $\frac{1}{4\beta}<1/2$.
    Since $F$ is allowed to block a flow of $1/2$ even in a good tree, it is possible to have  some edge $f\in F$ such that $f\in \+M(e)$.
\end{itemize}

\begin{figure}[H]
    \centering
    \subfigure[The red edge $e$ contains an edge in $F$.]{
        \label{fig:planA}
        \includegraphics[width=0.4\textwidth]{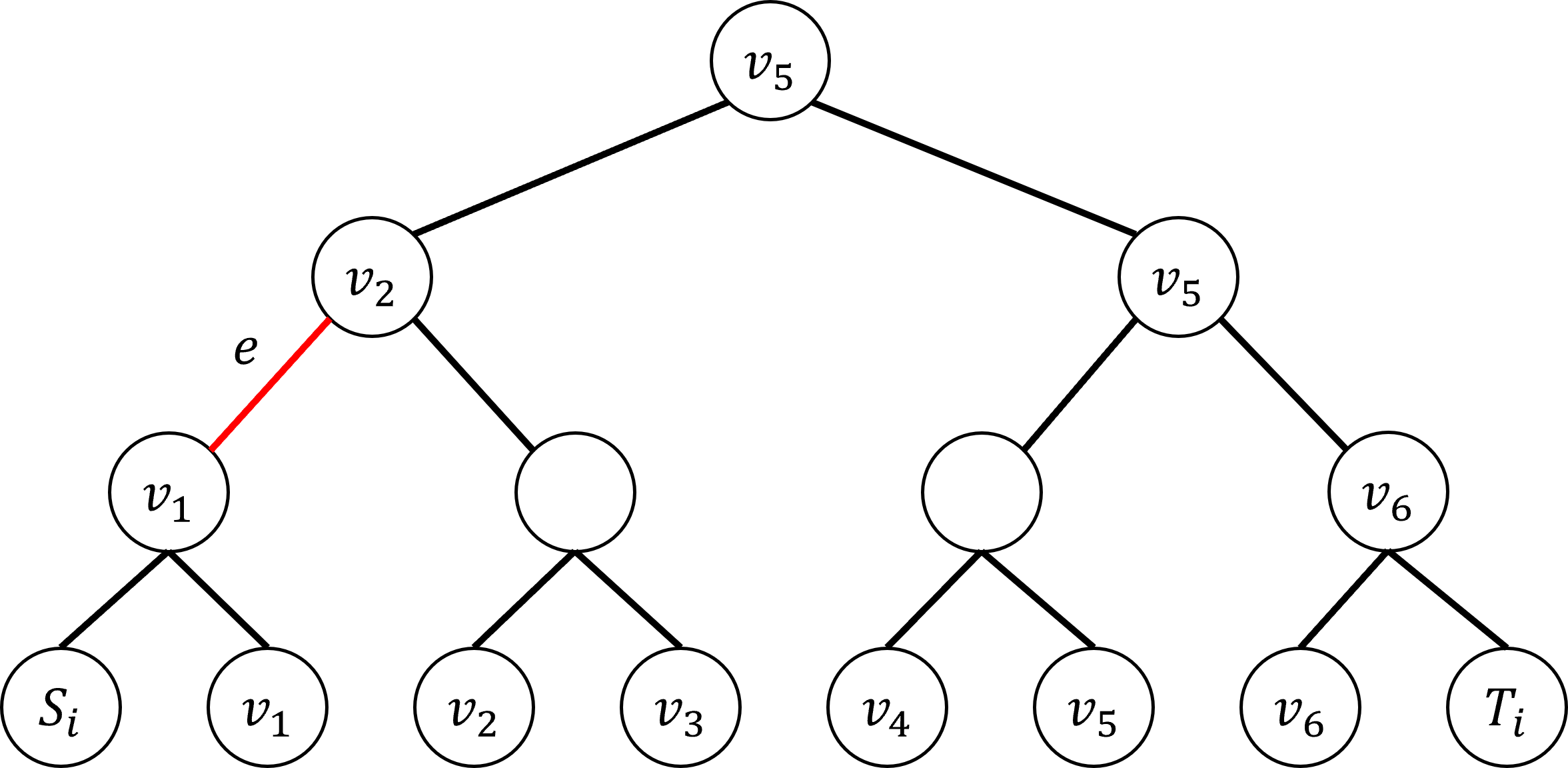}
    }
    \subfigure[A component-level path in blue. Red and green edges are crossing a component.]{
        \label{fig:planB}
        \includegraphics[width=0.4\textwidth]{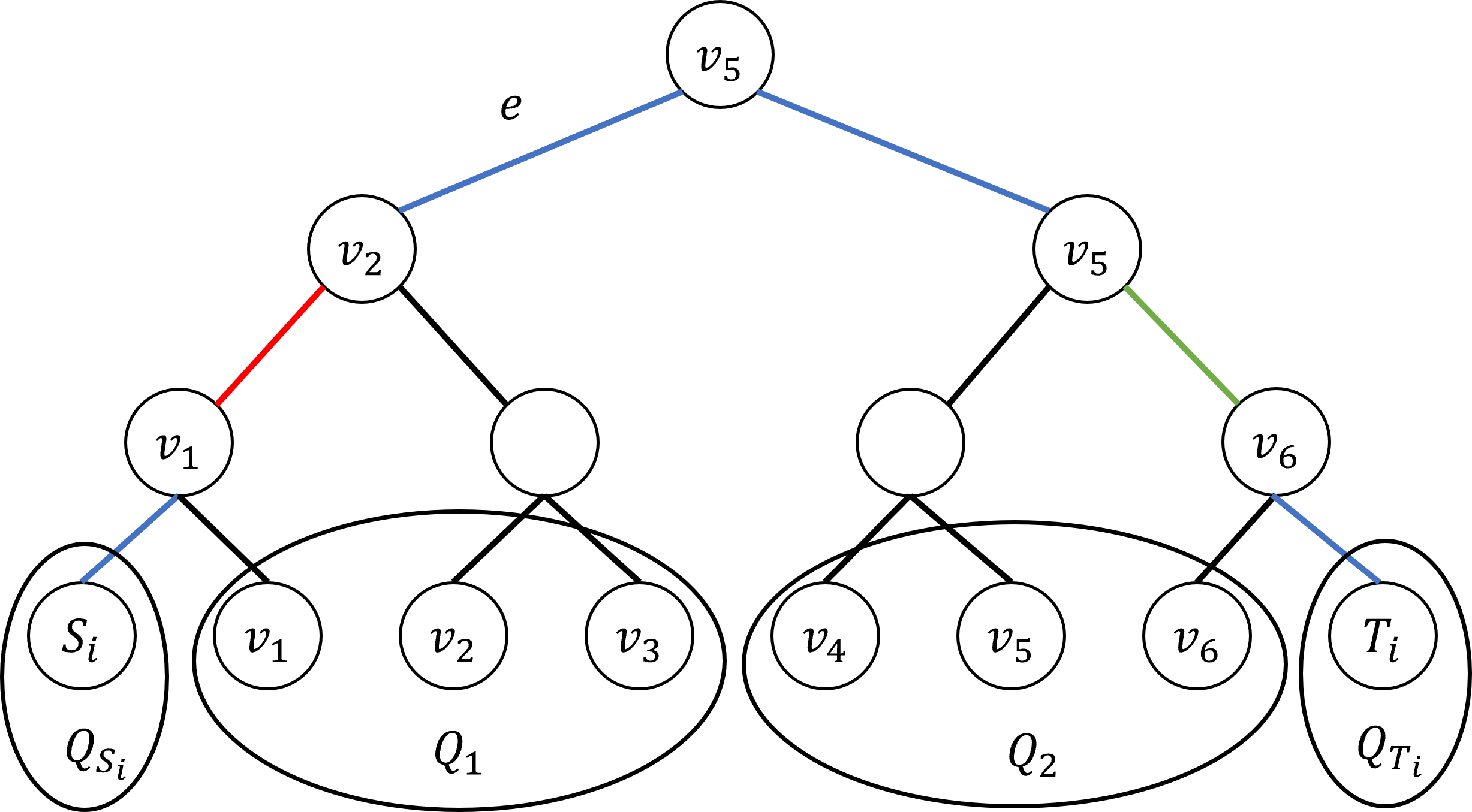}
    }
    \caption{A typical case for illustrating the analysis.}
    \label{fig:threeplan}
\end{figure}

To overcome the problem, we observe that the augmentation problem does not ask for a completely new path from $S_i$ to $T_i$. Let us move into a component-level viewpoint. The edge-set $F$ partitions the graph $H\setminus F$ into connected components. We know that every component $Q$ is already $(\ell+1)$-connected in $H$ so we do not need any new path inside $Q$. Let $Q_{S_i}$ the union of components that intersect with $S_i$, and similarly $Q_{T_i}$.  It means that a component-level path by edges in $G \setminus F$ that connects $Q_{S_i}$ and $Q_{T_i}$ suffices to show $(S_i,T_i)$ is $(\ell+1)$-connected. As a counterpart of the previous bad case, we present the reason why finding a component-level path is possible by the following example.

\begin{itemize}
    \item (Refer to \Cref{fig:planB})
    For the blue edge $e$ in \Cref{fig:planB}, we have $\widetilde{x_e}\geq 1$ because $\map(e)$ is a component-level path and it is in the unique path in the tree from $Q_{S_i}$ to $Q_{T_i}$.
    In this typical case, it is impossible to include any $f\in F$ in $M(e)$ in a good tree $\+T$ because otherwise it would block a unit-value (larger than $1/2$) flow.
\end{itemize}

The component-level path can be phrased in terms of the cut-based definition of tree-demand-pairs. Let $\mathbb{Q}$ be the set of connected components in $H \setminus F$. It is the same to say we need to connect every possible tree-demand-pairs defined as follows.
$$
Z = \{(A,B): (A,B) \text{ is a partition of $\mathbb{Q}$ such that $Q_{S_i}\in A$ and $Q_{T_i}\in B$} \}.
$$

However, this straightforward definition creates $2^{n}$ number of tree-demand-pairs, deriving an approximation ratio of $O(n)$. We discover an interesting technique to bound both the number and the flow of tree-demand-pairs.

\paragraph*{Connectivity by Shattered-Component-Level Paths.} We observe that the reason why we need the help of the component-level path is because of the existence of the red edge, which makes a component no longer connected on the tree if $F$ is removed. Referring to \Cref{fig:planB}, we call $Q_1$ a shattered component because of the red tree edge. On the contrary, assuming that the green tree edge does not include any edge in $F$, we call $Q_2$ an intact component. Because we can use the green edge to connect $Q_{S_i}$ and $Q_{T_i}$, we do not need to view $Q_2$ as an intermediate component in the component-level path. We consider a new component-level path that only contains shattered components, and we call it \emph{shattered-component-level paths}. Defining cut-based tree-demand-pairs on \emph{shattered-component-level paths} (Refer to \Cref{sec:generating-demand-pairs}.) significantly decreases the number of tree-demand-pairs while keeping the lower bound of flows. In conclusion, we prove the flowing properties in \Cref{sec:generating-demand-pairs}, \Cref{sec:algo:feasibility}, and \Cref{sec:algo:bound-demand-pairs}. We remark that \Cref{fig:threeplan} only illustrates a special case when $S_i$ and $T_i$ contain only one vertex. In general cases, we may have different paths between $S_i$ and $T_i$ on the tree. The edges between a vertex cut may not only be one single blue edge or green edge. We will formally discuss them later.

\begin{enumerate}
    \item (Refer to \Cref{sec:generating-demand-pairs}) If $(\+T,\+M,y)$ is good for $F$, we define the cut-based tree-demand-pairs by considering all possible shattered-component-level paths between all $(S_i,T_i)$.
    \item (Refer to \Cref{lem:dem-conn-F}) For each edge-set $F$, if we have sampled a good tree $\+T$ for $F$ and a subset of edges $E_b\subseteq E(\+T)\setminus \map^{-1}(F)$ such that all tree-demand-pairs is connected in $E_b$, then there is a path from $S_i$ to $T_i$ in $(E_\ell\bigcup \map(E_b)) \setminus F$ for each demand-pair $(S_i,T_i)$.
    \item (Refer to \Cref{lem:flow-available}) The flow in $\+T \setminus \+M^{\text{-}1}(F)$ between each tree-demand-pair is at least $\frac{1}{4\ell\beta}$.
    \item (Refer to \Cref{lem:num-demand}) Fix a tree $\+T$. The number of tree-demand-pairs defined by all $F$ and $(S_i,T_i)$ is at most $qn^{2\ell}\cdot2^{2\ell\beta}$. Note that it is $qn^{O(\ell)}$ if $\beta=O(\log n)$.
\end{enumerate}

Given the properties as above, we can show that our algorithm achieves the claimed approximation ratio.

\begin{lemma} [Connectivity]
\label{lem:alg_feasible}
Set $\tau_1= O(\ell \log n)$ and $\tau_2 = O(\log q + \ell (\log n + \beta))$. The output graph $H_{\ell+1}$ of \Cref{alg:main} is $(\ell+1)$-connected for each demand-pair $(S_i,T_i)$ with probability at least $1/2$.
\end{lemma}
\begin{proof}

If we can prove that $H_{\ell+1}\setminus F$ still connects $S_i$ with $T_i$ for all $F$ with $|F|=\ell$, and for all $i=1,...,q$, then $H_{\ell+1}$ is $(\ell+1)$-connected.
Since every $\+T$ in the support of $D$ has $\height(\+T)=O(\log n)$ by \Cref{lem:tree-height}, combining with \cref{lem:cgl}, we have a constant probability $\phi$ to connect a tree-demand-pair on $\+T$. When we sample a tree $\+T$, by Property 4 (\Cref{lem:num-demand}), the total number of tree-demand-pairs in $\+T$ is $\psi\leq qn^{2\ell}\cdot2^{2\ell\beta}$. Set 
$$
\tau_2 = \frac{1}{\phi} \cdot \log (4qn^{2\ell}\cdot2^{2\ell\beta})= O(\log q +  \ell (\log n + \beta)).
$$ 
We can connect every tree-demand-pairs on $\+T$ (denote this event by $\+C_{\+T}$) with probability:  
$$ 
    \Pr[\+C_{\+T}] \geq 1 - \psi \cdot (1-\phi)^{\tau_2} \geq 3/4.
$$  
Conditioned on $\+C_{\+T}$ for certain $\+T$ that is good for $F$, it holds that $H_{\ell+1}\setminus F$ is connected for all $(S_i,T_i)$ by Property 2 (\Cref{lem:dem-conn-F}). 
Next, because for every $F$, by \Cref{lem:good-tree-probability}, we have $1/2$ probability such that $\+T$ is good for $F$. Set 
$$
\tau_1= \frac{8}{3} \log 2n^{2\ell} = O(\ell \log n).
$$ 
Let $\psi'\leq n^{2\ell}$ be the number of possible edge-cut $F$. The probability that we sample a good tree $\+T$ for all $F$ and $\+T$ is fully connected ($\+C_{\+T}$) is at least:
$$
1 - \psi' \cdot (1-3/4\cdot1/2)^{\tau_1} \geq 1/2.
$$
It means that $H_{\ell+1} \setminus F$ is connected for every $(S_i,T_i)$ and for all $F$ with probability at least $1/2$, which concludes the lemma.
\end{proof}

\begin{lemma} [Cost]
\label{lem:alg_cost}
The output graph $H_{\ell+1}$ of \Cref{alg:main} has cost 
\[
O(\ell^2\beta^2\log^4 n(\ell\log n + \ell\beta + \log q) )\sum_{e\in E} c_e x_e.
\]
\end{lemma}
\begin{proof}
We use $4\ell \beta \cdot \sum_{e\in E} x_e = O(\ell \beta) \sum_{e\in E} x_e$ cost to buy the edges in $\LargeVal$. For the multiple rounds of TreeRounding, by \Cref{lem:cgl} and Property 3 (\Cref{lem:flow-available}), we know that for each round we pay
$$
O(\ell \beta^2 \log^3 n) \cdot \sum_{e\in E} c_e x_e.
$$
As $\tau_1=O(\ell\log n)$, $\tau_2=O(\log q + \ell (\log n + \beta))$, the total cost incurred by the algorithm is
\begin{align*}
   &O(\ell\log n\cdot(\log q + \ell(\log n+\beta)) \cdot \ell\beta^2 \cdot \log^3 n) \sum_{e\in E} c(e) x_e,
\end{align*}
as we claimed.
\end{proof}
\begin{corollary}
Combining \Cref{lem:alg_feasible} and \Cref{lem:alg_cost}, \Cref{alg:main} is a Monte Carlo $O(\ell^2\beta^2\log^4 n(\ell\log n + \ell\beta + \log q))$-approximation algorithm for the augmentation problem w.r.t. the LP-solution. It can be viewed as a Las Vegas one if we keep sampling until $H_{\ell+1}$ is $(\ell+1)$-connected.
\end{corollary}
Finally, by \Cref{lem:augmentation}, the corollary concludes an algorithm for the Group EC-SNDP with an approximation ratio of
$$
O(\beta^2 k^2\log k \log^4 n(k\log n + k\beta + \log q))
$$
as \Cref{thm:general-result} claims.
By using R\"acke's tree embedding where $\beta=O(\log n)$, we have the approximation ratio of
$$
O(k^2\log k \log^6 n(k\log n + \log q))
$$
as that in \Cref{thm:main}.

\subsection{Tree-Demand-Pairs for Shattered-Component-Level Paths}
\label{sec:generating-demand-pairs}

For a tree $\+T\in\+D$, let $\=Q^F$ be the set of all connected components in $H\setminus F$.
We remark that the subgraph $H$ might already consist of many connected components before removing the edges in $F$, for $V(H)=V(G)$ but $\abs{E(H)}$ might be much smaller than $\abs{E(G)}$.
For convenience, we think of each component $Q\in\=Q^F$ as a subset of $V$ instead of a subgraph of $H$.
For a nonempty set $S\subseteq V(G)$, suppose that $Q_1,Q_2,\ldots,Q_t$ are the components in $\=Q^F$ such that $Q_i\bigcap S\neq \emptyset$ for $1\le i\le t$, we define $Q_S=\bigcup_{i=1}^t Q_i$.

For a component $Q\in \=Q^F$, we say that $Q$ is \emph{shattered} if it is disconnected in $\+T\setminus \map^{-1}(F)$; otherwise, $Q$ is \emph{intact}.
For a cut $X\subseteq V(G)$, we say that $X$ is shattered if there exists some component $Q\in \=Q^F$ such that $X\bigcap Q\neq \emptyset$ and $(V(G)\setminus X)\bigcap Q\neq \emptyset$; otherwise, $X$ is intact.

Now we are ready to generate a collection $Z_F$ of tree-demand-pairs for each edge-set $F$ of $\ell$ edges and a good tree $\+T\in\calD$ for $F$. We generate all the possible partitions $(A',B')$ of shattered components in $\=Q^F$, denoted by $\mathbb{U}^F_{\+T}$.
Here we allow $A'=\emptyset$ or $B'=\emptyset$.
For each partition $(A',B')$, generate tree-demand-pairs $(A'\bigcup Q_{S_i}, B'\bigcup Q_{T_i})$ for each $i\in [q]$ such that $Q_{S_i}\bigcap Q_{T_i}=\emptyset$.
Formally, we have
\[
\textstyle Z_F := \set{(A'\bigcup Q_{S_i}, B'\bigcup Q_{T_i})\mid (A',B')\in \=U^F_{\+T}, i\in [q], Q_{S_i}\bigcap Q_{T_i}=\emptyset}.
\]
Notice that we only need $Z_F$ as tree-demand-pairs on tree $\+T$ when $\+T$ is \emph{good} for $F$, i.e., $F$ is \emph{good} on $\+T$. The set of all tree-demand-pairs for a tree $\+T$ is denoted by $Z_{\+T}=\bigcup_{F\subseteq E(H): \abs{F}=\ell, \text{$F$ is good on $\+T$}} Z_F$.



\subsection{Feasibility of the Reduction}
\label{sec:algo:feasibility}

We prove that the set of tree-demand-pairs of $Z_F$ on a tree $\+T$ corresponds to the problem of connecting all set pairs $(S_i,T_i)$ that are disconnected in $H\setminus F$.

\begin{lemma} \label{lem:dem-conn-F}
  For an edge-set $F\subseteq E(H)$ of $\ell$ edges, if there exist a tree $\+T$ and a set of edges $E_b\subseteq E(\+T)\setminus \map^{-1}(F)$ such that for every demand pair $(A,B)\in Z_F$ there is a path in $E_b$ connecting $A$ to $B$, then there is a $(S_i,T_i)$-path in $\tp{H\bigcup \map(E_b)}\setminus F$ for each $i\in [q]$. 
\end{lemma}

\begin{proof}
  Let $\overline{H}=H\bigcup \map(E_b)$ and $V=V(G)=V(H)$.
  We recall that, for a vertex set $S\subseteq V$, $Q_{S}$ is defined as the union of components in $\=Q^F$ intersecting with $S$, where $\=Q^F$ consists of all the connected components in $H\setminus F$ (a component is viewed as a subset of $V$ rather than a subgraph of $H$).
  If $Q_{S_i}\bigcap Q_{T_i}\neq \emptyset$, then we are done.
  Hence, we may assume that $Q_{S_i}\bigcap Q_{T_i}=\emptyset$.
  Then it suffices to show that, $\delta_{\overline{H}}(X)\ge 1$, for every cut $X\subseteq V$ satisfying $Q_{S_i}\subseteq X$ and $Q_{T_i}\subseteq V\setminus X$.

  If $X$ is a shattered cut, then there is some component $Q\in \=Q^F$ such that $Q\bigcap X\neq \emptyset$ and $Q\bigcap (V\setminus X)\neq \emptyset$.
  Thus, there is an edge $e=(u,v)\in E(H)\setminus F\subseteq E(\overline{H})$ such that $u\in Q\bigcap X$ and $v\in Q\bigcap (V\setminus X)$, implying that $\delta_{\overline H}(X)\ge 1$.

  If $X$ is an intact cut, then let $A=Q_{S_i}\bigcup \set{Q\in\=Q^F: Q\subseteq X}$ and $B=Q_{T_i}\bigcup \set{Q\in \=Q^F: Q\subseteq V\setminus X}$.
  Clearly, it holds that $(A,B)\in Z_F$.
  By the premise, there is a path $p$ in $E_b$ connecting $A\subseteq X$ to $B\subseteq V\setminus X$.
  Thus, it follows from the properties of R\"acke's tree that the corresponding path $\map(p)$ is in $H\bigcup \map(E_b)$.
  The premise $E_b\subseteq E(\+T)\setminus \map^{-1}(F)$ further ensures that $\map(p)$ is in $\tp{H\bigcup \map(E_b)}\setminus F$ because, otherwise, $p$ would use some edge in $\map^{-1}(F)$.
  Thus, $\delta_{\overline{H}}(X)\ge 1$, proving the lemma.
\end{proof}

\subsection{Bounding the Flow and the Number of Tree-Demand-Pairs}
\label{sec:algo:bound-demand-pairs}

As mentioned, we will later invoke a randomized algorithm to round the flows on all sampled trees $\+T$. However, given a specific tree $\+T$, it might not be feasible to connect the tree-demand-pairs for all the edge-sets.
This is because the maximum flow of some tree-demand-pairs may be too small.
To this end, let us recall the definition of ``good".

\begin{definition}
\label{def:good-tree}
We say that a tree $(\+T,M,y)$ is good for an edge-set $F\subseteq E(H)$, or $F$ is good on $\+T$, if $y(\map^{-1}(F)) \leq 1/2$.

\end{definition}

We recall that there are two types of cuts $X\subseteq V(H)$: intact and shattered.
We first prove that an intact cut has at least constant capacity with respect to $\widetilde{x}$ in the graph $G$.

\begin{lemma}
\label{lem:cap-intact-cut}
Let $X\subseteq V(G)$ be a cut such that $S_i\subseteq X$ and $T_i\subseteq V\setminus X$ for some set-pair $(S_i,T_i)$.
If $X$ is intact with respect to $F$, then $\widetilde{x}(\delta_{G}(X)\setminus F)\ge 3/4$.
\end{lemma}
\begin{proof}
Since $X$ is intact, all edges in $\delta_G(X)\setminus{F}$ are not in $H$, meaning that $\delta_G(X)\setminus F\subseteq\SmallVal$.
Thus, each edge $e$ in $\delta_G(X)\setminus F$ has capacity $\widetilde{x}_e=x_e$.
It follows from the LP constraints on $x$ that $x(\delta_G(X))\geq \ell+1$.
Since $F$ consists of $\ell$ edges, the capacity $x(F)$ is at most $\ell$.
Thus,
$$
\ell+1 \leq  x(\delta_{G}(X))
        =    x(\delta_{G}(X)\setminus{F}) + x(F)
        \leq x(\delta_{G}(X)\setminus{F}) + \ell
$$
Therefore, $x(\delta_G(X)\setminus F) \geq 1$, implying that
\[
  \widetilde{x}(\delta_{G}(X)\setminus{F})
  = x(\delta_{G}(X)\setminus{F}) - \sum_{e\in \delta_G(X)\setminus F\cmid x_e < (1/2n^2)\cdot(1/4\ell\beta)} x_e
  \ge 1 - n^2 \cdot (1/2n^2)\cdot (1/4\ell\beta)
  \ge 3/4. \qedhere
\]
\end{proof}

Since the capacity of an edge-set $F$ in any good tree is at most $2\beta\cdot\widetilde{x}(F)$, the next lemma follows.

\begin{lemma}[Lower Bound on the Flow Value in the Tree]
\label{lem:flow-available}
Let $F\subseteq E(H)$ be any edge-set of $\ell$ edges such that $Q_{S_i}\bigcap Q_{T_i}=\emptyset$ for some $i\in[q]$ and let $(A,B)$ be any tree-demand-pair generated for $S_i$ and $T_i$.
Then for a good tree $\+T$ for $F$ the flow that can be routed from $A$ to $B$ in $\+T\setminus\map^{-1}(F)$ is at least $1/(4\ell\beta)$.
\end{lemma}
\begin{proof}
  It is sufficient to prove that, for any $X\subseteq V$ such that $A\subseteq X$ and $B\subseteq V\setminus X$, the value of flow that can be routed in $\+T\setminus\map^{-1}(F)$ from $X$ to $V\setminus X$ is at least $1/(4\ell\beta)$.
  
  If $X$ is an intact cut, then $\widetilde{x}(\delta_G(X)\setminus F)\ge 3/4$ by \Cref{lem:cap-intact-cut} and $\flow^y_{\+T}(X,V\setminus X)\ge \flow_G^{\widetilde{x}}(X,V\setminus X)\ge 3/4$.
  Since $F$ is good on $\+T$, we have
  $
    y(\map^{-1}(F))\leq 1/2
  $
  which means that if the edges $\map^{-1}(F)$ are removed from $\+T$, then any flow in $\+T$ can decrease in value by at most $1/2$.
  Therefore, we have that
  $$\flow^y_{\+T\setminus \map^{-1}(F)}(X,V\setminus X)\ge \flow^y_{\+T}(X,V\setminus X)-1/2\ge 1/4\ge 1/(4\ell\beta).$$

  If $X$ is a shattered cut, then there is some component $Q\in \=Q^F$ such that $Q\bigcap X\neq \emptyset$ and $Q\bigcap (V\setminus X)\neq \emptyset$.
  The definition of tree-demand-pair $(A,B)$ implies that $Q$ cannot be a shattered component.
  Thus, $Q$ is still connected in $\+T\setminus \map^{-1}(F)$, and there is a path in $\+T\setminus \map^{-1}(F)$ connecting $X$ to $V\setminus X$.
  We claim that, for any edge $e\in E(\+T)$, if the partition $(Y,V\setminus Y)$ on leaves obtained by removing $e$ from $\+T$ satisfies $S_j\subseteq Y$ and $T_j\subseteq V\setminus Y$ for some $j\in [q]$, then $y_e\ge 1/(4\ell\beta)$.
  If so, then the value of the flow that can be routed from $X$ to $V\setminus X$ in $\+T\setminus \map^{-1}(F)$ is at least $1/(4\ell\beta)$.

  To prove the claim, note that $y_e=\widetilde{x}(\delta_G(Y))$.
  If there is some edge $e'$ in $\LargeVal$ crossing $Y$ and $V\setminus Y$, then $\widetilde{x}_{e'}=1/(4\ell\beta)$.
  Otherwise
  \begin{align*}
  \widetilde{x}(\delta_G(Y))
  &= \sum_{e\in\delta_G(Y)} x_e - \sum_{e\in\delta_G(Y)\cmid x_e< (1/2n^2)\cdot(1/4\ell\beta)} x_e \\
  &\ge (\ell+1) - n^2\cdot (1/2n^2)\cdot (1/4\ell\beta)  \\
  &\ge 1/(4\ell\beta).
  \qedhere
  \end{align*}
\end{proof}

\Cref{lem:flow-available} implies that for every edge-set $F\subseteq E(H)$ there is enough flow in a good tree to connect each tree-demand-pair $(A,B)\in Z_F$.
Next, we bound the number of tree-demand-pairs.

\begin{lemma}
\label{lem:num-shattered}
For any edge-set $F\subseteq E(H)$ of $\ell$ edges and a good tree $\+T$ for $F$, the number of shattered components in $\mathbb{Q}^F$ is at most $2\ell\beta$.
\end{lemma}
\begin{proof}
  Let $Q\in \=Q^F$ be any shattered component.
  By definition, there exist two vertices $s,t\in Q$ that are disconnected in $\+T\setminus \map^{-1}(F)$, while being connected in $H\setminus F$.
  Since each edge in $H\setminus F$ has capacity $1/(4\ell\beta)$, the maximum flow from $s_i$ to $t_i$ in $H\setminus F$ has value at least $1/(4\ell\beta)$ and so do these two vertices in the tree $T$.
  However, $s$ and $t$ are disconnected in $\+T\setminus \map^{-1}(F)$, meaning that there cannot be any flow between them in $\+T\setminus \map^{-1}(F)$.
  Applying this fact to every shattered component, we conclude that the total loss of flow value is at least $\nu/(4\ell\beta)$ after removing edges $\map^{-1}(F)$ from $\+T$, where $\nu$ is the number of shattered components.
  Therefore
  \begin{align*}
    \nu/(4\ell\beta)
    \le y(\map^{-1}(F))
    \leq 1/2.
  \end{align*}
  Hence we conclude that $\nu \le 2\ell\beta$.
\end{proof}

\begin{lemma}[Upper Bound on the Number of Tree-Demand-Pairs]
\label{lem:num-demand}
For each tree $\+T$, the number of tree-demand-pairs is $|Z_{\+T}|\le q\cdot n^{2\ell}\cdot 2^{2\ell\beta}$.
\end{lemma}

\begin{proof}
By \Cref{lem:num-shattered} we have
\[
  \abs{Z_{\+T}} \le \sum_{F\subseteq E(H)\cmid \abs{F}=\ell, \text{$F$ is good on $\+T$}} \abs{Z_F}
  \le \binom{\abs{E(H)}}{\ell} \cdot (2^{2\ell\beta}\cdot q)
  \le q\cdot n^{2\ell}\cdot 2^{2\ell\beta}.
  \qedhere
\]
\end{proof}





\section{Rounding Flows on the Tree}
\label{sec:roundtree}
For completeness, we illustrate how we apply the CGL's rounding in this section. The algorithm is shown in \Cref{alg:poly}.

\begin{algorithm}
\caption{TreeRounding}
\label{alg:poly}
\begin{algorithmic}[1]
        \For{$q=0,...,2\log(\frac{2n^2}{f})$ and $v\in V(\+T)$}
            \State With probability $\min\{1, \frac{8}{f}\cdot 2^{-q}\}$, $q$-mark $v$.
            \If{$v$ is $q$-marked}
                \State $y^{v,q}\gets 2^{q+1}y$ and $E^{v,q}\gets \emptyset$.
                \For{$\tau'$ interations}
                    \State $E^{v,q}\gets E^{v,q}\cup \text{RoundGKR}({\+T}_v,v,y^{v,q})$.
                \EndFor
                \State $H\gets H\cup\map(E^{v,q})$.
            \EndIf
        \EndFor
    \State return $H$.
\end{algorithmic}
\end{algorithm}

We prove the following \Cref{lem:cgl} by combining \Cref{lem:cgl-fea} and \Cref{lem:cgl-opt}.
\treerounding*

First, we introduce Garg-Konjevod-Ravi (GKR) Rounding which is repeatedly used in the TreeRounding algorithm. Denote the rounding algorithm for GST used in \cite{GargKR00} by RoundGKR. Let $\+T$ be the tree rooted at $r$, and $y$ be a fractional solution to the standard cut-based LP for GST. RoundGKR gives a way to connect every group $\+G_i$ to root $r$ with high probability. To be more precise, we state the following result implicitly shown by the authors.

\begin{lemma}
    \label{lem:gkr}
    Suppose that, for some $E'\subseteq E(\+T)$, capacity $y$ support an unit flow from $\+G_i\subseteq V(\+T)$ to $r$ in $E'$. RoundGKR$(\+T,r,y)$ where $\+T$ is a tree using fractional solution $y$ and $r$ is the root of $\+T$, gives a path $P\subseteq E'$ connects $\+G_i$ to $r$ with probability at least $\Omega(1/\log n)$.
\end{lemma}

In the TreeRounding algorithm, we have a tree embedding $\+T$ initially. For one iteration $q$ and node $v$, with some probability w.r.t $q$, it applies the subroutine RoundGKR for $\tau'$ iterations. In the end, the union of the corresponding edges in the original graph of the edges selected by the RoundGKR gives a solution to connect $A$ and $B$.

Now we are ready to prove \Cref{lem:cgl}. First, we show the probability that the algorithm connects each demand-pair $A$ and $B$.

\begin{lemma}
    \label{lem:cgl-fea}
    \Cref{alg:poly} connects $A$ to $B$ with constant probability $\phi$.
\end{lemma}

\begin{proof}
    We know that there is an $f$ unit of flow from $A$ to $B$ in $E'$. Then, we decompose this flow into a family of flow paths. By discarding all the flow paths whose flow is less than $\frac{1}{2n^2} \cdot f$, we obtain a new family of flow path $P$. Denote the flow of $P_j$ by $p_j$. Let $\mu_v$ be the total amount of flow turning at $v$. Then we have
    $$
        \sum_{v\in V(\+T)}\mu_v = \sum_{j=1}^{|P|}p_j\geq f - \frac{1}{2n^2} \cdot f \cdot |E(T)| \geq \frac{f}{2}.
    $$
    For each node $v \in V(\+T)$, let $q_v\in\set{0,...,2\log(\frac{2n^2}{f})}$ be an integer such that $\mu_v \in (2^{-q_v-1},2^{-q_v}]$.

    Denote the event that some node $v\in V(\+T)$ is $q_v$-marked by $\+I$. We claim that $\Prob[\widebar{\+I}|\+G]\leq e^{-3/4}$. Let $X_v$ be 1 if node $v\in V(\+T)$ is $q_v$-marked and 0 otherwise. We have
$$
    \Prob\sqtp{\widebar{\+I}|\+G}=\Prob\sqtp{\sum_{v\in V(T)}X_v=0}.
$$
Since $\set{X_v}$ are independent,
$$
    \Exp\sqtp{\sum_{v\in V(\+T)}X_v}=\sum_{v\in V(\+T)}\frac{8}{f}\cdot2^{-q_v}\geq \frac{8}{f}\sum_{v\in V(\+T)}p_v\geq 4.
$$
By Chernoff's bound,
$$
    \Prob\sqtp{\sum_{v\in V(\+T)}X_v\leq 1}\leq e^{-\frac{1}{3}(\frac{3}{4})^2 4}=e^{-3/4}.
$$

Conditioned on $\+I$, there is some node $v\in V(\+T)$ that is $q_v$-marked. Consider the event $\+C$ that $A$ and $B$ are connected by the union of the solution computed by RoundGKR in $\tau'$ iterations on node $v$ for $q=q_v$. Observe that $2^{q_v+1}\phi_v\geq 1$. Therefore, by \Cref{lem:gkr}, RoundGKR selects a correct path with probability at least $1-(1-\Omega(\frac{1}{\log n}))^{\tau'}$. Thus, $\Prob[\widebar{\+C}|\+G,\+I]\leq \eps, $ where $\tau'=O(\log n)$. Altogether, the probability that the Tree Rounding connects $A$ to $B$ is at least
\begin{align*}
    1-\Prob[\widebar{\+G}]-\Prob[\widebar{\+I}|\+G]-\Prob[\+G]&\cdot\Prob[\+I|\+G]\cdot\Prob[\widebar{\+C}|\+G,\+I] 
     \geq 1-1/2-e^{-3/4}-\eps=:\phi>0.\qedhere
\end{align*}

\end{proof}

Next we analyze the cost incurred by the algorithm.

\begin{lemma}
\label{lem:cgl-opt}
$\Exp[c(H)]=O(\frac{1}{f} \cdot \beta \cdot \log^3 n) \sum_{e\in E}c_e x_e$.
\end{lemma}

\begin{proof}
First, an edge $e'\in E(\+T)$ is selected iff $e'$ is in $E^{v,q}$ for some $q$ and some $q$-marked $v\in V(\+T)$. For each iteration of RoundGKR, we select $e'$ with probability $y^{v,q}=2^{q+1} y$. To proceed these iterations, the vertex $v$ need to be $q$-marked. Thus, the corresponding probability is at most $\frac{8}{f}\cdot 2^{-q}$. Putting everything together, edge $e$ is selected with probability at most
$$
\height(\+T)\cdot\sum_{q}\frac{8}{f}\cdot2^{-q}\cdot2^{q+1}y(e')\cdot\tau'=O\tp{\frac{\log^3 n}{f}}y(e').
$$
Then, an edge $e\in E$ is selected iff any edge in $\+M^{-1}(e)$ is selected by RoundGKR. Therefore, the expectation of an edge $e$ is selected is at most
$$
    O\tp{\frac{\log^3 n}{f}}\cdot\Exp\sqtp{\sum_{e'\in\map^{-1}(e)}y(e')}\leq O\tp{\frac{1}{f} \cdot \beta \cdot \log^3 n}\cdot x_e.
$$

Taking the summation of all edges, the expected cost of $c(H)$ is at most
\[
    O\tp{\frac{1}{f}\cdot \beta \cdot \log^3 n}\sum_{e\in E}c(e)x_e. \qedhere
\]
\end{proof}

\section{Conclusion and Open Problems}
\label{sec:conclusion}

In this paper, we have presented an approximation algorithm for Group EC-SNDP whose approximation ratio is $O(k^3\,\polylog(k,n,q))$. It is quite interesting that this factor resemblances the approximation ratio of $O(k^3\log n)$ for VC-SNDP \cite{ChuzhoyK12}.
%
The factor $k$ appears quite naturally in network design problems and might be the approximability threshold.
However, the best known negative result still has a lower bound of $k^{1/5-\epsilon}$, for $\epsilon > 0$, assuming $\NP\neq\ZPP$, \cite{ChalermsookDELV18} (combined with \cite{Lae2014} and the improvement in \cite{Man2019}). 
While it is rather unnatural that a survivable network design problem would have a hardness factor beyond $k$, the recent result of Liao, Chen, Laekhanukit and Zhang \cite{LiaoCLZ2022} shows that this is the case for the sister problem of Group EC-SNDP, namely the {\em $k$-connected directed Steiner tree} problem ($k$-DST) (a.k.a, {\em directed single-source $k$-connectivity}).
They showed that the approximation lower bound of $k$-DST is, indeed, at least $\Omega(2^{k/2}/k)$, which is far beyond polynomial. 
Any results in both directions would be surprisingly interesting, either the existence of $k^{1-\epsilon}\polylog(n)$-approximation algorithms or a hardness threshold of $O(k^{1+\epsilon})$, for some $\epsilon>0$. 

We are aware of the capacity-based probabilistic tree embedding for {\em $\alpha$-balanced} graphs, which appears in the work of Ene, Miller, Pachocki and Sidford \cite{EneMPS16}. In our opinion, there is a high chance that the probabilistic capacity mapping would work for us, thus generalizing our framework to directed graphs, or more specifically, to the case of $k$-DST on $\alpha$-balanced graphs. Unfortunately, the construction in \cite{EneMPS16} is tailored for minimizing congestion for single-source oblivious routing, and it is not clear whether the bound holds for other set-pairs in the graph.
Our algorithm, on the other hand, requires the congestion guarantee to hold for pairwise subsets, in which one contains the source (i.e., root). 
Thus, one direction to push forward on studying survivable network design is in developing a capacity-based probabilistic tree-embedding that is able to deal with the more general settings of oblivious routing on $\alpha$-balanced graphs.
\paragraph*{Acknowledgement.}

This work is supported by Science and Technology Innovation 2030 –“New Generation of Artificial Intelligence” Major Project No.(2018AAA0100903), NSFC grant 61932002, Program for Innovative Research Team of Shanghai University of Finance and Economics (IRTSHUFE) and the Fundamental Research Funds for the Central Universities.

Qingyun Chen is supported in part by NSF grants CCF-2121745 and CCF-1844939. Bundit Laekhanukit is partially supported by the 1000 Talents Plan award by the Chinese government.

\bibliography{ref}
\bibliographystyle{alpha}
\end{document}